\documentclass[journal]{IEEEtran}
\ifCLASSOPTIONcompsoc
\else
\fi
\ifCLASSINFOpdf
\else
\fi
\usepackage{graphicx}
\usepackage{amssymb}
\usepackage{amsmath}
\usepackage{bm}
\usepackage[Algorithm, boxed]{algorithm}
\usepackage{algorithmic}
\usepackage[bookmarks=false]{hyperref}
\usepackage{breakurl}
\usepackage{multirow}
\newtheorem{theorem}{Theorem}

\newtheorem{remark}{Remark}

\begin{document}
%
\title{A New Approach of Deriving Bounds between Entropy and Error 
from Joint Distribution: Case Study for Binary Classifications}
%
%
%
%

\author{Bao-Gang Hu,~\IEEEmembership{Senior Member,~IEEE},~Hong-Jie Xing
\IEEEcompsocitemizethanks{\IEEEcompsocthanksitem Bao-Gang Hu is with NLPR/LIAMA,
Institute of Automation, Chinese Academy of Sciences, Beijing 100190, China.\protect\\
E-mail: hubg@nlpr.ia.ac.cn \protect\\
   Hong-Jie Xing is with College of Mathematics and Computer Science, HeBei University,
Baoding, 071002, China. \protect\\
E-mail: hjxing@hbu.edu.cn 
}
\thanks{}}

\IEEEcompsoctitleabstractindextext{%
\begin{abstract}
The existing upper and lower bounds between entropy and error  
are mostly derived through an 
inequality means without linking to joint distributions.
In fact, from either theoretical or application viewpoint, 
there exists a need to achieve a complete set of interpretations to the
bounds in relation to joint distributions.
For this reason, in this work 
we propose a new approach of deriving the bounds between entropy and error 
from a joint distribution. The specific case study is given on 
binary classifications, which can justify the need of 
the proposed approach. 
Two basic types of classification errors are investigated, namely, 
the Bayesian and non-Bayesian errors. For both errors, 
we derive the closed-form
expressions of upper bound and lower bound 
in relation to joint distributions.  
The solutions
show that Fano's lower bound is an exact bound for any type 
of errors in a relation diagram of ``Error Probability vs. 
Conditional Entropy''.
A new upper bound for the Bayesian error is derived with respect to 
the minimum prior probability, which is generally tighter 
than Kovalevskij's upper bound.

\end{abstract}

\begin{keywords}
Entropy, error probability, Bayesian errors, analytical, upper bound, lower bound
\end{keywords}}

\maketitle

\IEEEdisplaynotcompsoctitleabstractindextext

%
\IEEEpeerreviewmaketitle

\section{Introduction}
\label{sec:introduction}
In information theory, the relations between entropy and error 
probability are one of the important fundamentals.  
Among the related studies, one milestone is Fano's inequality (also known as Fano's 
lower bound on the error probability of decoders),
which was originally proposed in 1952 by Fano,
but formally published in 1961 \cite{Fano1961}.
It is well known that Fano's inequality plays a critical role in 
deriving other theorems and criteria in information theory 
\cite{cover2006}\cite{verdu1998}\cite{yeung2002}. 
However, within the research community,  it has not been widely accepted exactly 
who was
first to develop the upper bound on the error probability \cite{Golic1999}. 
According to \cite{Vajda2007} \cite{Morales2010},
Kovalevskij \cite{Kovalevskij1965} was recognized as the first 
to derive the upper bound of the error probability in relation to entropy
in 1965. Later, several
researchers, such as
Chu and Chueh in 1966 \cite{Chu1966}, Tebbe and Dwyer III 
in 1968 \cite{Tebbe1968},
Hellman and Raviv in 1970 \cite{Hellman1970},
independently developed upper bounds.    

The upper and lower bounds of error probability 
have been a long-standing topic in studies on
information theory \cite{Chen1971} \cite{Bassat1978} \cite{Golic1987} \cite{Feder1994}
\cite{Han1994} \cite{Poor1995} \cite{Harremoes2001} 
\cite{Erdogmus2004}\cite{Vajda2007} \cite{Morales2010}\cite{Ho2010}. 
However, we consider two issues that have received less 
attention in these studies:

\begin{itemize}
\item[I.] What are the closed-form relations between each bound and 
joint distributions in a diagram of entropy and error probability?   %
\item[II.] What are the lower and upper bounds in terms of the non-Bayesian errors 
if a non-Bayesian rule is applied in the information processing? %
\end{itemize}

The first issue implies a need for a complete set 
of interpretations to the bounds in relation to joint distributions,
so that both error probability and its error components are known
for interpretations. 
We will discuss the reasons of the need in the later sections of this paper.  
Up to now, most existing studies derived the bounds through an 
inequality means without using joint distribution information.
Therefore, their bounds are not described by a 
generic relation to joint distributions. 
Using the truncated-distribution approach, 
a significant study by Ho and Verd\'{u} \cite{Ho2010}
was reported recently on established the relations
for general cases of variables with 
finite alphabets and countably infinite alphabets.
Regarding the second issue, to our best knowledge, it seems that
no study is shown in open literature 
on the bounds in terms of the non-Bayesian errors.
We will define the Bayesian and non-Bayesian errors in Section III.
The non-Bayesian errors are also of importance because
most classifications are realized
within this category.

The issues above form the motivation behind this work.
We take binary classifications as a problem
background since it is more common and understandable
from our daily-life experiences. 
Moreover, we intend to simplify settings within a binary state
and Shannon entropy definitions for a case study from 
an expectation that the central 
principle of the approach is well highlighted by simple examples. 
The novel contribution of
the present work is given from the following three aspects:

\begin{itemize}
\item[I.] A new approach is proposed for deriving bounds 
directly through the optimization process based on a 
joint distribution, which is significantly different from all other
existing approaches. One advantage of using the approach 
is a possible solution of closed-form expressions
to the bounds.
\item[II.] A new upper bound in a diagram of 
``Error Probability vs. Conditional Entropy''
for the Bayesian errors is derived with a closed-form
expression in the binary state, which is not reported before. 
The new bound is generally tighter than Kovalevskij's upper bound.%
\item[III.] The comparison study on the bounds in terms of the Bayesian and non-Bayesian errors 
are made in the binary state. The connections of bounds are explored for a first time 
between two types of errors. %
\end{itemize}

In the first aspect, we also conduct the actual derivation 
using a symbolic software tool, 
which presents a standard and comprehensive solution 
in the approach.
The rest of this paper is organized as follows. In Section II, we
present related works on the bounds. For a problem background
of binary classifications, several related definitions 
are given in Section III. 
The bounds are given and discussed for the Bayesian
and non-Bayesian errors in Sections IV and V, respectively.
Interpretations to some key points are presented in 
Section VI. Finally, in Section VII we conclude the work and present
some discussions.
The source code from using symbolic software for 
the derivation is included in Appendixes  A and B.

\section{Related Works}
Two important bounds are introduced first, which form the baselines for the
comparisons with the new bounds. 
They were both derived from inequality conditions\cite{Fano1961}\cite{Kovalevskij1965}. 
Suppose the random variables $X$ and $Y$ representing input and output messages 
(out of $m$ possible messages), 
and the conditional entropy $H(X|Y)$ representing the average amount of information lost
on $X$ when given $Y$.
Fano's lower bound \cite{Fano1961} is given in a form of:
\setcounter {equation} {1}
$$ H(X|Y) \leq  H(P_e) + P_e {log_2 (m-1)},  \eqno (1) $$
where $P_e$ is the {\it error probability} (sometimes, 
also called {\it error rate}
or {\it error} for short), and $H(P_e)$ is the binary 
entropy function defined by \cite{shannon1948}:
$$ H(P_e) = - P_e log_2 P_e - (1- P_e)log_2 (1-P_e). \eqno (2) $$
The base of the logarithm is 2 so that the units are {\it bits}.

The upper bound is given by Kovalevskij \cite{Kovalevskij1965}
in a piecewise linear form \cite{Tebbe1968}:
\setcounter {equation} {2}
\begin{equation}
\label{equ:}
\begin{array}{r@{\quad}l}
&  H(X|Y) \geq  log_2 k +k(k+1)(log_2 \frac{k+1}{k})(P_e - \frac{k-1}{k}), \\
& ~~~~~~~~~~~~~~~~~~~~~~~~~~~~~~~~~~~~ and ~ ~k < m,~ m \geq 2,
\end{array} 
\end{equation}
where $k$ is a positive integer number, but defined to be smaller than $m$.
For a binary classification ($m=2$), Fano-Kovalevskij bounds become:
$$ H^{-1}(P_e) \leq P_e \leq \frac{H(X|Y)}{2}, \eqno (4) $$
where $H^{-1}(P_e)$ is an inverse of $H(P_e)$.
Feder and Merhav \cite{Feder1992} depicted bounds of eq. (4)
and presented interpretations on the two specific points from 
the background of data compression problems. 

Studies from the different perspectives have been reported on the bounds 
between error probability and entropy. The initial
difference is made from the entropy definitions, such as Shannon
entropy in \cite{Chen1971}\cite{Golic1987}\cite{wang2009}\cite{hu2012},
and R\'{e}nyi entropy in \cite{Feder1994}\cite{Vajda2007}\cite{Morales2010}.  
The second difference is the selection of bound relations, such
as ``$P_e$ vs. $H(X|Y)$'' in \cite{Chen1971}\cite{Feder1992}, ``$H(X|Y)$ vs. $P_e $'' 
in \cite{Golic1987} \cite{Feder1994}\cite{Vajda2007}\cite{Morales2010}\cite{Ho2010},
``$P_e$ vs. $MI(X,Y)$'' in \cite{Eriksson2005}\cite{Fisher2009}, and
``$NMI(X,Y)$ vs. $A$'' in \cite{wang2009}, 
where $A$ is the accuracy rate, $MI(X,Y)$ and $NMI(X,Y)$ are the mutual information and
normalized mutual information between variables $X$ and $Y$, respectively.
Another important study is made on the tightness of 
bounds. Several investigations \cite{Taneja1985} \cite{Poor1995}
\cite{Erdogmus2004} \cite{Ho2010} have been reported on
the improvement of bound tightness. Recently, a study in \cite{hu2012}
suggested that an upper bound from the Bayesian errors should be 
added, which is generally neglected in the bound analysis. 

\section{Binary Classifications and Related Definitions}

Classifications can be viewed as one component in pattern 
recognition systems \cite{duda2001}. 
Fig. 1 shows a schematic diagram of the pattern recognition systems. 
The first unit in the
systems is termed {\it representation} in the present problem background, 
but called {\it encoder} in communication background. 
This unit processes the tasks of {\it feature selection}, 
or {\it feature extraction}.
The second unit is called {\it classification} or {\it classifier} in applications.
Three sets of variables are involved in the systems, namely, {\it target
variable} $T$, {\it feature variables} $X$, and 
{\it prediction variable} $Y$. While $T$ and $Y$ are
univariate discrete random variables for representing labels of the samples,
$X$ can be high-dimension random variables either in forms of 
discrete, continuous, or their combinations. 
\begin{figure}
\centering
\includegraphics[width=3.3in]{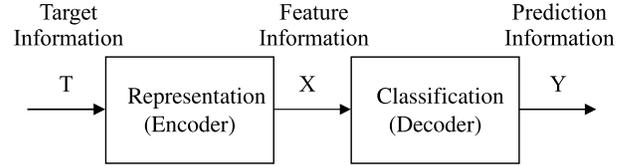}
\caption{Schematic diagram of the pattern recognition systems 
(modifications on FIGURE 1.7 in \cite{duda2001}).   
}
\label{fig:1a}
\end{figure}

In this work, binary classifications are considered as a case study 
because they are more fundamental in applications. 
Sometimes, multiclass classifications 
are processed by binary classifiers \cite{Cristianini2000}.
In this section, we will present several necessary definitions 
for the present case study. 
Let $\textbf{x}$ be a 
random sample satisfying $\textbf{x} \in \mathcal{X} \subset {R^d}$, which is 
in a $d$-dimensional feature space and will be classified. 
The true (or target) state $t$ of $\textbf{x}$ is within the finite set of 
two classes, $t \in \mathcal{T} = \{t_{1}, t_{2}\}$, 
and the prediction (or output) state $y=f(\textbf{x})$ is 
within the two classes, $y \in \mathcal{Y} = \{y_{1}, y_{2}\}$, 
where $f$ is a function for classifications.
Let $p(t_{i})$ be the {\it prior probability} of class $t_{i}$ 
and $p(\textbf{x}|t_{i})$ be the {\it conditional probability density function}
(or {\it conditional probability}) of $\textbf{x}$ given that it 
belongs to class $t_{i}$.

{\it Definition 1:} {\it (Bayesian error in binary classification)} 
In a binary classification, the {\it Bayesian error}, denoted by $P_e$, is defined 
by \cite{duda2001}:
\setcounter {equation} {4}
\begin{equation}
 \label{equ:}
P_e=\int \limits_{R_2} 
p(t_1|\textbf{x})p(t_1)d\textbf{x}
+ \int \limits_{R_1} p(t_2|\textbf{x})p(t_2)d\textbf{x},
\end{equation}
where $R_i$ is the {\it decision region} for class $t_{i}$. The two regions 
are determined by the Bayesian rule:
\begin{equation}
\label{equ:}
 \begin{array}{r@{\quad}l}
& \mathrm{Decide} ~ R_1 ~ \mathrm{if} ~ \dfrac{p(\textbf{x}|t_1)p(t_1)}{p(\textbf{x}|t_2)p(t_2)} \ge 1,\\
& \mathrm{Decide} ~ R_2 ~ \mathrm{if} ~ \dfrac{p(\textbf{x}|t_1)p(t_1)}{p(\textbf{x}|t_2)p(t_2)} < 1,
\end{array} 
\end{equation}
In statistical classifications, the Bayesian error 
is the {\it theoretically lowest} probability of error \cite{duda2001}.  

{\it Definition 2:} {\it (Non-Bayesian error)} 
The {\it non-Bayesian error}, denoted by $P_E$, is defined to be 
any error which is larger than the Bayesian error, that is:
\begin{equation}
 \label{equ:}
P_E > P_e,
\end{equation}
for the given information of $p(t_{i})$ and $p(\textbf{x}|t_{i})$.

\begin{remark}
Based on the definitions above, for the given joint distribution
the Bayesian error is unique, but
the non-Bayesian errors are multiple. Fig. 2 shows the 
Bayesian {\it decision boundary}, $x_b$, on a
univariate feature variable $x$ for equal priors. The 
Bayesian error is $P_e= e_1+e_2$.
Any other decision boundary different from $x_b$
will generate the non-Bayesian error
for $P_E>P_e$. 
\end{remark}

\begin{figure}
\centering
\includegraphics[width=3.4in,height=1.9in]{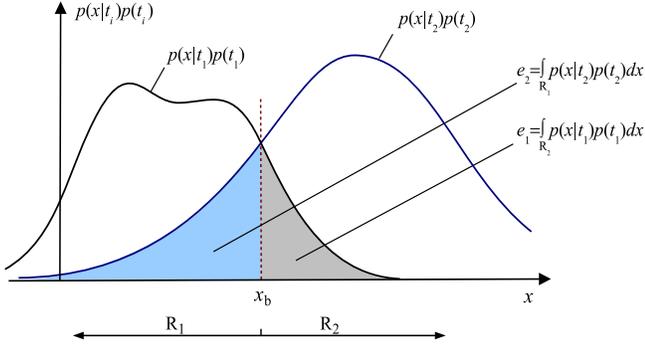}
\caption{ Bayesian decision boundary $x_b$ for equal priors $p(t_i)$
in a binary classification
(modifications on FIGURE 2.17 in \cite{duda2001}).   
}
\label{fig:1b}
\end{figure}

In a binary classification, the {\it joint distribution},
$p(t,y)=p(t=t_i,y=y_j)=p_{ij}$,
is given in a general form of:
$$\setlength\arraycolsep{0.1em}
 \begin{array}{r@{\quad}l}
  & p_{11}=p_1-e_1, ~  p_{12}=e_1,\\
   & p_{21}=e_2, ~~~~~~~  p_{22}=p_2-e_2, 
 \end{array}
\eqno (8) 
$$
where $p_1=p(t_1)$ and $p_2=p(t_2)$ are the prior 
probabilities of Class 1 and Class 2, respectively;
their associated errors 
(also called {\it error components})
are denoted by $e_1$ and $e_2$.
Fig. 3 shows a graphic diagram of the probability transformation
between target variable $T$ and perdition variable $Y$ 
via their joint distribution $p(t,y)$ in a binary
classification. 
The constraints in eq. (8) are given by \cite{duda2001}:
$$\setlength\arraycolsep{0.1em}
 \begin{array}{r@{\quad}l}
& 0 < p_1 <1, ~~ 0 < p_2 <1, ~ p_1+p_2=1 \\
& 0 \leq e_1 \leq p_1, ~ 0 \leq e_2 \leq p_2. 
\end{array} 
\eqno (9) 
$$

\begin{figure}
\centering
\includegraphics[width=2.0in]{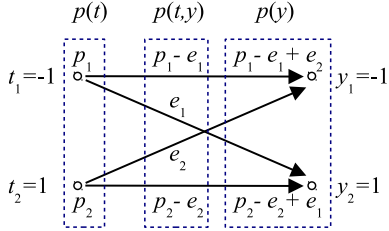}
\caption{Graphic diagram of the probability transformation
between variables $T$ and $Y$ 
in a binary classification.
}
\label{fig:1c}
\end{figure}

In this work, we use $e$ to denote error probability, or error variable, 
for representing either the Bayesian error or non-Bayesian error. 
They are calculated from the same formula:
$$ e (P_e, ~ or~ P_E)= e_{1}+e_{2}. \eqno (10) $$

{\it Definition 3:} {\it (Minimum and maximum error bounds in 
binary classifications)} 
Classifications suggest the minimum error bound as:
$$ (P_E)_{min} = (P_e)_{min} = 0, \eqno (11) $$
where the subscript {\it min} denotes the minimum value.
The maximum error bound for the Bayesian error in binary classifications is 
\cite{hu2012}:
$$ (P_e)_{max} = p_{min} =min \{p_1,p_2 \}, \eqno (12) $$
where the symbol {\it min} denotes 
a {\it minimum} operation.
For the non-Bayesian error, its maximum error bound
becomes 
$$ (P_E)_{max} = 1. \eqno (13) $$

\begin{remark}
For a given set of joint distributions in the bound studies, 
one may fail to tell if it
is the solution from using 
the Bayesian rule or not. For simplification, 
we distinguish the set to be one for the Bayesian errors if an error 
rate $e$ always satisfies 
the relation of $e \le p_{min}$.
Otherwise, it is a set for the non-Bayesian errors. 
\end{remark}

In a binary classification, the
{\it conditional entropy},  $H(T|Y)$, is calculated
from the joint distribution in (8):
$$\setlength\arraycolsep{0.1em}
 \begin{array}{r@{\quad}l}
 H(T|Y)  = & H(T)-MI(T,Y)\\
 = & -p_{1}log_2 p_1-p_{2}log_2 p_2 \\
& -{e_1}log_2 \frac{e_1}{(p_2+e_1-e_2)p_1} \\
& -{e_2}log_2  \frac{e_2}{(p_1-e_1+e_2)p_2} \\ 
& -{(p_1-e_1)}log_2  \frac{(p_1-e_1)}{(p_1-e_1+e_2)p_1} \\ 
& -{(p_2-e_2)}log_2  \frac{(p_2-e_2)}{(p_2+e_1-e_2)p_2}, ~~
\end{array} 
\eqno (14) 
$$
where $H(T)$ is a {\it binary entropy} of the random variable $T$,
and $MI(T, Y)$ is 
{\it mutual information} between variables $T$ and $Y$.

\begin{remark}
When a joint distribution $p(t,y)$ is given, its associated
conditional entropy $H(T|Y)$ is uniquely determined. 
However, for the given $H(T|Y)$, it is generally unable to reach
a unique solution to $p(t,y)$, but mostly multiple solutions shown
later in this work.
\end{remark}

{\it Definition 4:} {\it (Admissible point, admissible set, and their properties
in diagram of entropy and error probability)} 
In a given diagram of entropy and error probability, 
if a point in the diagram is possibly to be realized from 
a non-empty set of joint distributions for the given classification information,
it is defined to be an {\it admissible point}. 
Otherwise, it is a {\it non-admissible point}.
All  admissible points will form an {\it admissible set}
(or {\it admissible region(s)}),
which is enclosed by the bounds (also called {\it boundary}).
If every point located on the boundary is admissible (or non-admissible), 
we call this admissible set {\it closed} (or {\it open}). 
If only a partial portion of boundary points   
is admissible, the set is said {\it partially closed}. 
For an admissible point with the given conditions, 
if it is realized only by a unique joint 
distribution, it is called a {\it one-to-one mapping} point. If more
than one joint distribution is associated to the same admissible point,
it is called  a {\it one-to-many mapping} point. 

We consider that classifications present
an exemplary justification of raising the first issue
in Section I about the bound studies.
The main reason behind the issue is that a single index 
of error probability may not be sufficient for  
dealing with classification problems. 
For example, when processing class-imbalance problems 
\cite{He2009}\cite{Sun2009},
we need to distinguish
{\it error types}. In other words, 
for the same error probability $e$ (or even the
same admissible point), we are required to know the error
components of $e_1$ and $e_2$ as well. 
Suppose one encounters a medical diagnosis problem, where $p_1$ generally represents
the {\it majority class} for {\it healthy} persons
(labeled with {\it negative} or  {-1} in Fig. 3), and $p_2$ the {\it minority
class} for {\it abnormal} persons (labeled with {\it positive} or  {1}). 
A class-imbalance problem
is then formed. While $e_1$ (also called 
{\it type I error} ) is tolerable, $e_2$ (or {\it type II error}) seems
intolerable because abnormal persons are considered to be ``{\it healthy}''.
Hence, from either theoretical or application viewpoint, 
it is necessary for establishing relations between bounds and joint distributions,
which can provide error type information within error probability 
for better interpretations to the bounds.   

\section{Upper and lower bounds for Bayesian errors}

In this work, we select the bound relations between
entropy and error probability. Furthermore, The bounds 
and their associated error components 
are also given by the following two theorems
in a context of binary classifications.

\begin{theorem} {\it (Lower bound and 
associated error components)}
The lower bound in a diagram of ``$P_e$ vs. $H(T|Y)$'' 
and the associated error components are given by:
\begin{equation}\tag{15a}
\label{equ:15a}
P_e \geq min \{ 0, G_1(H(T|Y)) \},
\end{equation}
\begin{equation}\tag{15b}
\label{equ:15b}
\begin{array}{r@{\quad}l}
for ~ G_1^{-1}(P_e) & = H(T|Y)\\
& = - P_e log_2 P_e - (1- P_e)log_2 (1-P_e),\\
P_e   & =e_1+e_2 \le p_{min},
\end{array} 
\end{equation}

\begin{equation}\tag{15c}
(e_1, e_2) = \left\{ \begin{array}{r@{}l}
& (0.5,0) ~~ or ~~ (0, 0.5), ~~~~~~~~~  if ~  P_e   = 0.5,\\ \nonumber
& (\frac {P_e(1-p_1-P_e)} {1-2P_e},\frac {P_e(p_1-P_e)} {1-2P_e}), 
~~~ otherwise, \nonumber
\end{array} \right.
\end{equation}
where $H(T|Y)$ is the conditional entropy of 
of $T$ when given $Y$,
and $G_1$ is called the {\it lower bound function} 
(or {\it lower bound}). However, one can only achieve the closed-form solution on its inverse 
function, $G_1^{-1} (\cdot)$, not on itself.  
\end{theorem}

\begin{proof}
Based on eq. (14), the lower bound function is derived from the following 
definition:
$$\setlength\arraycolsep{0.1em}
 \begin{array}{r@{\quad}l}
& G_1^{-1}(e)  = arg \max \limits_{e}H(T|Y),\\
& ~~ subject ~ to ~ eqs. ~ (9) ~and~ (10), 
\end{array} 
\eqno (16) 
$$
where we take $e$ for the 
input variable in the derivations. 
Eq. (16) describes the function of the maximum 
$H(T|Y)$ with respect to $e$, and the function
needs to satisfy the general constraints of joint 
distributions in eq. (9). 
$H(T|Y)$ seems to be governed by the four variables
from $p_i$ and $e_i$ in eq. (14). 
However, only two independent parameter variables determine the 
solutions of (14) and (16). The variable reduction 
from four to two is due to the two specific
constrains imposed between parameters,  
that is, $p_1+p_2=1$ and $e_1+e_2=e$.
When we set $p_1$ and $e_1$ as two independent variables,
eq. (16) is then equivalent to solving the following problem:
$$\setlength\arraycolsep{0.1em}
 \begin{array}{r@{\quad}l}
& G_1^{-1}(p_1, e_1) = arg \max \limits_{e=P_e}H(T|Y),\\
& ~~ subject ~ to ~ eqs. ~ (9) ~and~ (10). 
\end{array} 
\eqno (17) 
$$
$G_1^{-1}(p_1, e_1)$ is a continuous and differentiable 
function with respect to the two variables. 
A differential approach is applied analytically for searching 
the {\it critical points}
of the optimizations in eq. (17). We achieve the two differential equations 
below and set them to be zeros:
\begin{equation}\tag{18}
\left\{ \begin{array}{r@{}l}
& \frac {\partial H(T|Y)} {\partial e_1}= 
log_2 \frac {(p_1-e_1)(P_e-e_1)(1+2e_1-p_1-P_e)^2} 
{e_1(1+e_1-p_1-P_e)(p_1+P_e-2e_1)^2} =0,
\\ \nonumber
& \frac {\partial H(T|Y)} {\partial p_1}= log_2 \frac {(p_1-2e_1+P_e)(1+e_1-p_1-P_e)} 
{((p_1-e_1)(1+2e_1-p_1-P_e))} =0 . \nonumber
\end{array} \right.
\end{equation}
By solving them simultaneously, we obtain the three
pairs of the critical points through analytical derivations:  
\begin{equation}\tag{19a}
\left\{ \begin{array}{r@{}l}
& e_1= \frac {P_e(1-p_1-P_e)} {1-2P_e} , 
\\ \nonumber
& p_1= 
\frac {P_e+2e_1P_e-e_1-P_e^2} {P_e},
~~~~~~~~~~~~~~~~~~~~~~~~~~~~~
~~~ ~~ \nonumber
\end{array} \right.
\end{equation}
\begin{equation}\tag{19b}
\left\{ \begin{array}{r@{}l}
& e_1= \frac {p_1(p_1+P_e-1)} {2P_1-1} , 
\\ \nonumber
& p_1= 
\frac {1-P_e} {2} +e_1+\frac {1} {2} 
\sqrt{1+P_e^2+4e_1^2-4e_1P_e-2P_e}, \nonumber
\end{array} \right.
\end{equation}
\begin{equation}\tag{19c}
\left\{ \begin{array}{r@{}l}
& e_1= \frac {p_1(p_1+P_e-1)} {2P_1-1} , 
\\ \nonumber
& p_1= 
\frac {1-P_e} {2} +e_1-\frac {1} {2} 
\sqrt{1+P_e^2+4e_1^2-4e_1P_e-2P_e}. \nonumber
\end{array} \right.
\end{equation}
The highest order of each variable, $e_1$ and $p_1$, in eq. (18) is 
four. However, we can see the component 
within the first function in eq. (18), $(1+2e_1-p_1-P_e)^2$, 
will degenerate the total solution order from four to three.  
Therefore, the three pairs of critical points exhibit a complete 
set of {\it possible solutions} to the 
problem in eq. (17). The {\it final solution} should be the pair(s) that satisfies 
both the maximum 
$H(T|Y)$ with respect to $e_1$ for the
given $e=P_e$ and the constraints.
Due to high complexity of the nonlinearity
of the second-order partial differential equations on $H(T|Y)$, 
it seems intractable to examine the three pairs 
analytically for the final solution.  

To overcome the difficulty above, we apply a symbolic software tool, 
Maple\texttrademark 9.5
(a registered trademark of Waterloo Maple, Inc.), 
for a {\it semi-analytical} solution to the problem (see Maple code 
in Appendix A).  
For simplicity and without loss of generality in classifications, 
we consider $p_1$ and $P_e$ are known constants in the function. 
The concavity property of $H(T |Y )$ with respect to $e_1$ 
in the ranges defined in eq. (9)
is confirmed numerically by varying data on $p_1$ and $P_e$. 
A single maximum solution on $H(T|Y)$ is always obtained, 
but it is described by the two sets of $e_1$ 
in (19) alternatively in different conditions 
of $p_1$ and $P_e$.
\end{proof}

\begin{remark}
Theorem 1 achieves the same lower bound found 
by Fano \cite{Fano1961} (Fig. 4), which is general for finite alphabets
(or multiclass classifications). One specific relation to 
Fano's bound is given by the {\it marginal 
probability} (see eq. (2-144) in \cite{cover2006}):
\begin{equation}\tag{20}
\label{equ:20}
\begin{array}{r@{\quad}l}
p(y)= (1-P_e, \frac {P_e} {m-1}, ..., \frac {P_e} {m-1}),
\end{array} 
\end{equation}
which is termed {\it sharp} for attaining equality 
in eq. (1) \cite{cover2006}. 
We call Fano's bound an {\it exact} lower bound because every
point on it is sharp.
The sharp conditions in terms of error components 
in (15c) are a special case of the study in \cite{Ho2010}, 
and can be derived directly from their Theorem 1. 
\end{remark}

\begin{theorem} {\it (Upper bound and associated error components)}
The upper bound and
the associated error components are given by:
\begin{equation}\tag{21a}
\label{equ:21a}
P_e \leq min \{ p_{min}, G_2(H(T|Y)) \},
\end{equation}
\begin{equation}\tag{21b}
\label{equ:21b}
\begin{array}{r@{\quad}l}
for ~ G_2^{-1}(e) & = H(T|Y)\\
& =-p_{min}log_2 \frac{p_{min}}{P_e+p_{min}}-P_elog_2 
\frac{P_e}{P_e+p_{min}}, ~~
\end{array} 
\end{equation}
\begin{equation}\tag{21c}
\label{equ:21c}
\begin{array}{r@{\quad}l}
~~~~~ and  ~ P_e  &  =e_1+e_2 \le p_{min}, \\
 ~ e_i &  = p_{j}, ~ e_j= 0, ~ p_i \geq p_j, ~ i\neq j, ~ i,j=1,2
\end{array} 
\end{equation}
where $G_2$ is called the {\it upper bound function}
(or {\it upper bound}).
Again, the closed-form solution can be achieved only on its inverse 
function of $G_2^{-1} (\cdot)$.
\end{theorem}

\begin{proof}
The upper bound function is obtained from
solving the following equation:
$$\setlength\arraycolsep{0.1em}
 \begin{array}{r@{\quad}l}
& G_2^{-1}(p_1, e_1) = arg \min \limits_{e=P_e}H(T|Y),\\
& ~~ subject ~ to ~ eqs. ~ (9) ~and~ (10). 
\end{array} 
\eqno (22) 
$$
Because the concavity
property holds for $H(T |Y )$ with respect to $e_1$ 
for the constraints defined
in eq. (9), the possible solutions of $e_1$ should be located at
the two ending points of its feasible range, $(0, P_e)$.
We can take the point which produces the smaller $H(T |Y )$
as the final solution.
The solution from Maple code shown in Appendix B confirms 
the closed-form expressions in eq. (21).
\end{proof}

\begin{remark}
Theorem 2 describes a novel set of upper bounds which is in general 
{\it tighter} 
than Kovalevskij's bound \cite{Kovalevskij1965} for binary 
classifications (Fig. 4). 
For example, when $p_{min}=0.2$ is given, the upper bounds
defined in eq. (21) shows
a curve ``$O-C$''  plus a  line ``$C-C'$''.
Kovalevskij's upper bound, given by a line 
``$O-C-A$'', is sharp only at Point $O$ and Point $C$.
The solution in eq. (21c) confirms an 
advantage of using the proposed optimization 
approach in derivations so that a closed-form expression 
of the exact bound is possibly achieved.
\end{remark}

In comparison, Kovalevskij's upper bound described in eq. (3) is 
general for multiclass classifications.
This bound misses a general relation to error components
like eq. (21c), although the relation is restricted to a binary state.
For distinguishing from the Kovalevskij's upper bound,
we also call $G_2$ a {\it curved upper bound}. 
The new {\it linear upper bound}, $(P_e)_{max}=p_{min}$, 
shows the maximum error for the Bayesian
decisions in binary classifications \cite{hu2012},
which is also equivalent to the solution of a blind guess
when using the maximum-likelihood decision \cite{duda2001}. 
If $p_1 = p_2$, the upper bound becomes a single curved one.

\begin{figure}
\centering
\includegraphics[width=3.3in]{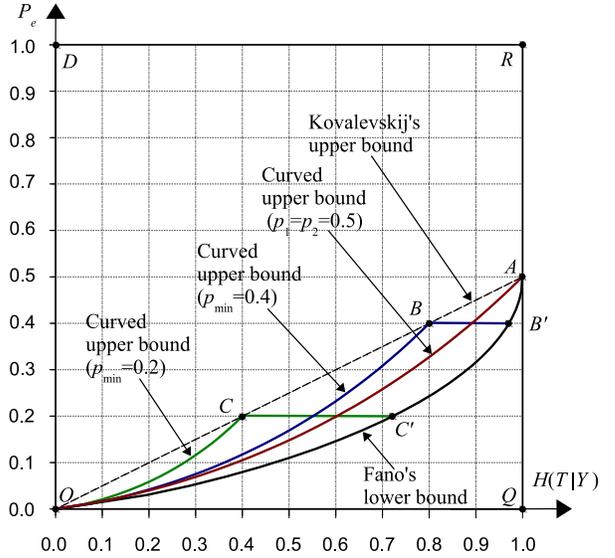}
\caption{Plot of bounds in a ``$P_e$ vs. $H(T|Y)$'' diagram.
}
\label{fig:4}
\end{figure}

\begin{remark}
The lower and upper bounds defined by eqs. (15)  and (21) 
form a closed admissible region in the diagram of ``$P_e$ vs. $H(X|Y)$''. 
The shape of the admissible region
changes depending on a single parameter of $p_{min}$.
\end{remark}

\section{Upper and lower bounds for non-Bayesian errors}

In classification problems, the Bayesian errors
can be realized only if one has the exact
information about all probability distributions of classes. 
The assumption above is generally impossible
in real applications. 
In addition, various classifiers are designed by employing the non-Bayesian rules,
such as the conventional decision trees, artificial neural networks 
and supporting vector machines \cite{duda2001}. 
Therefore, 
the analysis of the non-Bayesian errors presents
significant interests in classification studies.  

\begin{figure*}[!t]
\begin{equation}\tag{24c}
(e_1, e_2) = \left\{ \begin{array}{r@{}l}
& (0.5,0) ~~ or ~~ (0, 0.5), ~~~~~~~~~~~~~  if ~  p_1= p_2 = P_E   = 0.5,\\ \nonumber
& (\frac {P_E(1-p_1-P_E)} {1-2P_E},\frac {P_E(p_1-P_E)} {1-2P_E}), 
~~~~~ if ~ (1-p_1-P_E)(p_1-P_E)(P_E-0.5)>0,\\ \nonumber
& (\frac {p_1(p_1+P_E-1)} {2P_1-1},\frac {(1-p_1)(p_1-P_E)} {2P_1-1}) , 
~~ otherwise, \nonumber
\end{array} \right.
\end{equation}
\end{figure*} 

{\it Definition 5:} {\it (Label-switching in binary classifications)} 
In  binary classifications, a label-switching operation
is an exchange between two labels.
Suppose the original joint distribution 
is denoted by: 
\begin{equation}\tag{23a}
\label{equ:23a}
\begin{array}{r@{\quad}l}
p_A(t,y): & p_{11}=a, ~  p_{12}=b,\\
& p_{21}=c, ~  p_{22}=d. 
\end{array} 
\end{equation}
A label-switching operation will change 
the prediction labels in Fig. 3 to be
$y_1=1$ and $y_2=-1$, and generate the following
joint distribution:
\begin{equation}\tag{23b}
\label{equ:23b}
\begin{array}{r@{\quad}l}
p_B(t,y): & p_{11}=b, ~  p_{12}=a,\\
& p_{21}=d, ~  p_{22}=c. 
\end{array} 
\end{equation} 

{\it Proposition 1:} {\it (Invariant property from label-switching)} 
The related entropy measures, including $H(T)$, 
$H(Y)$, $MI(T,Y)$, and $H(T|Y)$, 
will be invariant to labels, or unchanged 
from a label-switching operation in binary classifications.  
However, the error $e$ will be changed to be $1-e$.
\begin{proof}
Substituting the two sets of joint distributions in eq. (23) into each entropy measure 
formula respectively, one can obtain the same results. The error change is obvious.  
\end{proof}

\begin{theorem} {\it (Lower bound and upper bound for non-Bayesian error
without information of $p_1$ and $p_2$ )}
In a context of binary classifications, 
when information about $p_1$ and $p_2$ is unknown (say, before classifications),
the lower bound and upper bound for the non-Bayesian error
are given by:
\begin{equation}\tag{24a}
\label{equ:24a}
G_1(H(T|Y))\leq P_E \leq 1- G_1(H(T|Y)),
\end{equation}
\begin{equation}\tag{24b}
\label{equ:24b}
\begin{array}{r@{\quad}l}
for ~ G_1^{-1}(P_E) & =H(T|Y)\\
& =- P_E log_2 P_E - (1- P_E)log_2 (1-P_E), \\
P_E   & =e_1+e_2 \le 1,
\end{array} 
\end{equation}
\begin{equation}\tag{24c}
\label{equ:24c}
\begin{array}{r@{\quad}l}
 (see ~ the ~top ~of ~this ~page) ~~ ~~~~~~
\end{array} 
\end{equation}
where we call the upper bound in eq. (24a), $1- G_1(H(T|Y))$,
the {\it general upper bound} (or {\it mirrored lower bound}),
which is a mirror of Fano's lower bound with 
the mirror axis along $P_E=0.5$. 
Both bounds share the same expression for calculating
the associated error components in eq. (24c). 
When $P_E \le 0.5 $, their components, $e_1$ and $e_2$, 
correspond to the lower bound, 
otherwise, to the upper bound.  
\end{theorem}

\begin{proof}
Suppose an admissible point is located at
the lower bound which shows $P_E \le 0.5 $.
By a label-switching operation, one can obtain
the mirrored 
admissible point at $1-P_E \ge 0.5 $, which is 
located at the mirrored lower bound. 
Proposition 1 suggests both points share the same value of $H(T|Y)$.
Because $P_E$ is the smallest one for the given conditional entropy $H(T|Y)$,
its mirrored point is the biggest one for creating the general upper bound.
\end{proof}

\begin{remark}
Fano's lower bound, its mirror bound, and the 
axis of $P_E$ 
form an admissible region, denoted by a boundary 
``$O-F'-A-F-D-O$'' in Fig. 5, for the non-Bayesian error
when information about $p_1$ and $p_2$ is unknown. 
On the axis of $P_E$, only Points O and D are admissible. Hence,
the admissible region is partially closed.   
\end{remark}

\begin{theorem} {\it (Admissible region(s) for non-Bayesian error
with known information of $p_1$ and $p_2$)}
In binary classifications,
when information about $p_1$ and $p_2$ is known,
a closed admissible region for the non-Bayesian error is 
generally
formed (Fig. 5) by Fano's lower bound, the general upper bound, 
the curved upper bound $G_2^{-1}(\cdot)$,
the {\it mirrored upper bound} of $G_2^{-1}(\cdot)$, and the upper bound $H(T|Y)_{max}$. 
For the $H(T|Y)_{max}$ bound,
its associated error components are given by:
\begin{equation}\tag{25}
\begin{array}{r@{}l}
~for ~ H(T|Y)= H(T|Y)_{max} =H(e=p_{min}),  ~~~~~ ~~~~~  ~~~~~ ~~~
\\ \nonumber
(e_1,e_2)=
\left\{ \begin{array}{r}
(0.25, 0.25),~ ~~~~~ ~~~~~ ~
if ~ p_1= p_2=P_E = 0.5, \\ \nonumber
(\frac {p_1(1-p_1-P_E)} {1-2p_1},\frac {P_E(1-p_1)-p_1(1-p_1)} {1-2p_1}), 
~ otherwise. \nonumber
\end{array} \right.
\end{array} 
\end{equation}
\end{theorem}

\begin{proof}
Following the proof in Theorem 3, one can get the mirrored 
upper bound of $G_2^{-1}(\cdot)$.
The upper bound $H(T|Y)_{max}$ is calculated 
from the condition of $H(T|Y) \leq H(T)$ \cite{cover2006}. 
For the given $p_1$ and $p_2$, $H(T|Y)_{max}$ is a constant.
Because $H(T|Y)_{max}$ also implies a minimization of $MI(T,Y)$ in eq. (14),  
its associated error components can be obtained from the 
minimization relation of $MI(T,Y)$ in forms of 
(see eq. (35) in \cite{hu2008}): 
\begin{equation}\tag{26}
\label{equ:26}
 \frac {p_{11}} {p_{21}} = \frac {p_{12}}  {p_{22}}. 
\end{equation} 
\end{proof}

\begin{figure}
\centering
\includegraphics[width=3.3in]{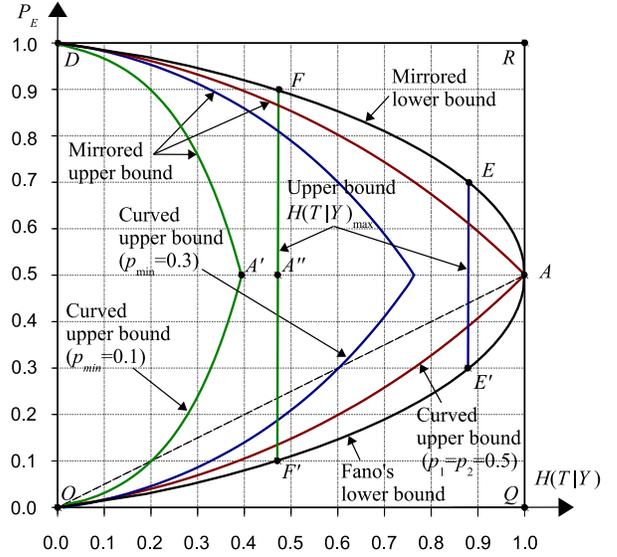}
\caption{Plot of bounds in a ``$P_E$ vs. $H(T|Y)$'' diagram.
}
\label{fig:2}
\end{figure}

\begin{remark}
Eqs. (25) and (26) equivalently imply a zero value for 
the mutual information, $MI(T,Y)=0$, 
which suggests {\it no correlation} \cite{duda2001} or
{\it statistically independent} \cite{cover2006}
between two variables $T$ and $Y$.
\end{remark}

\begin{remark}
When information of $p_1$ and $p_2$ is known, 
the shape of the admissible region(s) is fully dependent 
on a single parameter $p_{min}$.
Two closed admissible regions are formed only when $p_1 = p_2$ (Fig. 5). 
One region is from   
Fano's lower bound and the upper bound. 
The other is from the mirrored upper bound
and the general upper bound.  
In general, the non-Bayesian error $P_E$ can be 
higher than Kovalevskij's bound.
\end{remark}

\section{Classification Interpretations to some key points}

For better understanding the theoretical results
from a background of classifications,
interpretations are given to
some key points shown in Figs. 4 and 5, respectively.
Those key points may hold special features in classifications. 

{\it Point O}: 
This point represents a zero value of $H(T|Y)$. 
It also suggests a {\it perfect classification} without
any error ($P_e=P_E=0$) by a specific setting 
of the joint distribution:
$$\setlength\arraycolsep{0.1em}
 \begin{array}{r@{\quad}l}
& p_{11}=p_1, ~  p_{12}=0,  \\
& p_{21}=0, ~~  p_{22}=p_2.
\end{array} 
\eqno (27) 
$$
This point is always admissible and independent 
of error types.

{\it Point A}: 
This point shows the maximum ranges of $H(T|Y)=1$ for {\it class-balanced}
classifications ($p_1=p_2$). Three specific classification settings 
can be obtained for representing this point. The two settings 
from eq. (24c) are
actually {\it no classification}: 
$$\setlength\arraycolsep{0.1em}
 \begin{array}{r@{\quad}l}
& p_{11}=1/2, ~  p_{12}=0,  ~or ~~  p_{11}=0, ~  p_{12}=1/2,\\
& p_{21}=1/2, ~  p_{22}=0, ~~~~~~ p_{21}=0, ~  p_{22}=1/2.
\end{array} 
\eqno (28) 
$$
They also indicate {\it zero information} \cite{mackay2003}
from the classification decisions. 
The other setting is a {\it random guessing} from eq. (25):
$$\setlength\arraycolsep{0.1em}
 \begin{array}{r@{\quad}l}
& p_{11}=1/4, ~  p_{12}=1/4,  \\
& p_{21}=1/4, ~  p_{22}=1/4. 
\end{array} 
\eqno (29) 
$$
For the Bayesian errors, this point is always included by
both Fanos' bound and Kovalevskij's bound.
However, according to the upper bounds
defined in (21a), this point is non-admissible
whenever the relation of $p_1=p_2$ does not hold.
For the non-Bayesian errors, the point is either
admissible or non-admissible depending on the 
given information about $p_1$ and $p_2$. 
This example suggests that the admissible property
of a point should generally rely on the given information
in classifications.

{\it Point D}: 
This point occurs for the non-Bayesian classifications 
in a form of:
$$\setlength\arraycolsep{0.1em}
 \begin{array}{r@{\quad}l}
& p_{11}=0, ~~  p_{12}=p_1,  \\
& p_{21}=p_2, ~  p_{22}=0.
\end{array} 
\eqno (30) 
$$
In this case, one can exchange the labels for a perfect classification.

{\it Point B}: 
This point is located at the corner formed by 
the curved and linear  upper bounds,
with $H(T|Y)=0.8$ and $e=0.4$. 
In apart from Point $O$, this is another point obtained from eq. (21)
that sets at Kovalevskij's upper bound.
The point can be realized from either Bayesian or 
non-Bayesian classifications. 
Suppose $p_1 > p_2=0.4$ for the Bayesian classifications. 
One will achieve Point $B$ by a classification:
$$\setlength\arraycolsep{0.1em}
\begin{array}{r@{\quad}l}
& p_{11}=0.2, ~  p_{12}=0.4,  \\
& p_{21}=0, ~~~  p_{22}=0.4,
\end{array} 
\eqno (31) 
$$ 
for a one-to-one mapping. In other words,
the point becomes non-admissible whenever 
$p_{min} \neq 0.4$. 
If the non-Bayesian errors are considered,
this point will 
possess a one-to-many mapping. For example, one can get 
another setting from solving $H(p_{min})=0.8$ for
$p_{min}$ first. Then, by substituting the relations of
$p_2=p_{min}$ and $P_E=0.4$ into eq. (25), one can get 
the error components. The numerical 
results show the 
approximation solutions with $p_{min}\approx 0.2430, ~
e_1\approx0.2312$, and $e_2\approx0.1688$ for another setting
of Point $B$.

{\it Point $B'$}: 
The point located at the lower bound, like Point $B'$, will produce 
a one-to-many mapping for either the Bayesian errors 
or non-Bayesian errors.
One specific setting 
in terms of the Bayesian errors is:
$$\setlength\arraycolsep{0.1em}
\begin{array}{r@{\quad}l}
& p_{11}=0.6, ~  p_{12}=0,\\
& p_{21}=0.4, ~  p_{22}=0, 
\end{array} 
\eqno (32) 
$$
which suggests zero information from classifications.
More settings can be obtained from eq. (15).
For example, if given $p_1=0.55$, $p_2=0.45$ and $P_e=0.4$,
one can have:
$$\setlength\arraycolsep{0.1em}
 \begin{array}{r@{\quad}l}
& p_{11}=0.45, ~  p_{12}=0.1,\\
& p_{21}=0.3, ~~ p_{22}=0.15.
\end{array} 
\eqno (33) 
$$
The non-Bayesian errors will enlarge the set of 
one-to-many mapping for 
an admissible point of the Bayesian errors
due to the relaxed condition of (13). 
One setting is for the balenced error 
components: 
$$\setlength\arraycolsep{0.1em}
 \begin{array}{r@{\quad}l}
& p_{11}=0.3, ~  p_{12}=0.2,\\
& p_{21}=0.2, ~  p_{22}=0.3.
\end{array} 
\eqno (34) 
$$
Eq. (24c) will be applicable for 
deriving a specific setting
when $p_1$ and $P_E$ are given. 
For example, two settings 
can be obtained:
$$\setlength\arraycolsep{0.1em}
 \begin{array}{r@{\quad}l}
~if ~~~&  p_1=0.25, ~~  P_E=0.4,~\\
then & e_1=0.175, ~  e_2=0.225, 
\end{array} 
\eqno (35) 
$$
$$\setlength\arraycolsep{0.1em}
 \begin{array}{r@{\quad}l}
~if ~~~&  p_1=0.3, ~~~~  P_E=0.4,\\
then & e_1=0.225, ~  e_2=0.175.
\end{array} 
\eqno (36) 
$$
for representing the same point, Point $B'$, which is 
located at $H(T|Y)\approx0.9710$ and $P_E=0.4$
in the diagram (Fig. 4).   

{\it Points E and $E'$}: 
All points located at the general upper bound, like Point $E$, will
correspond to the settings from the 
non-Bayesian errors. 
If a point located at the lower bound, say $E'$, it
can represent settings from either the Bayesian or
non-Bayesian errors depending on the given information
in classifications.  
Points $E$ and $E'$ form the mirrored points. Their settings 
can be connected by a relation in (23), but not a necessary.
For example, 
one specific setting for Point $E'$ with $p_1=0.3$ and $p_2=0.7$ is: 
$$\setlength\arraycolsep{0.1em}
 \begin{array}{r@{\quad}l}
& p_{11}=0, ~ p_{12}=0.3,  \\
& p_{21}=0, ~ p_{22}=0.7,
\end{array} 
\eqno (37) 
$$
the other for Point $E$ with $p_1=0.8$ and $p_2=0.2$ is: 
$$\setlength\arraycolsep{0.1em}
 \begin{array}{r@{\quad}l}
& p_{11}=\frac {20} {30}, ~ p_{12}=\frac {4} {30},  \\
& p_{21}=\frac {5} {30}, ~ p_{22}=\frac {1} {30}.
\end{array} 
\eqno (38) 
$$
They are mirrored to each other but have no label-switching relation.

{\it Points $A'$ and $A''$}: 
When $P_E=0.5$ and $p_{min}=0.1$, Points $A'$ and $A''$
form a pair as the ending points for the given conditions.
Supposing $p_1=0.9$ and $p_2 = 0.1$, one can get 
the specific setting for Point $A'$ from eq. (21c): 
$$\setlength\arraycolsep{0.1em}
 \begin{array}{r@{\quad}l}
& p_{11}=0.4, ~ p_{12}=0.5,  \\
& p_{21}=0, ~~~ p_{22}=0.4,
\end{array} 
\eqno (39) 
$$
and one for Point $A''$ from eq. (25):
$$\setlength\arraycolsep{0.1em}
 \begin{array}{r@{\quad}l}
& p_{11}=0.45, ~ p_{12}=0.45,  \\
& p_{21}=0.05, ~ p_{22}=0.05.
\end{array} 
\eqno (40) 
$$

{\it Points Q and R}: 
The two points are specific due to their positions in the diagrams. 
For either type of errors, both points are
non-admissible in the diagrams, because no setting exists
in binary classifications which can represent the points.

\section{Summary and discussions}
\label{sec:6}

This work investigates into upper and lower bounds between
entropy and error probability. An optimization approach
is proposed to the derivations of the bound 
functions from a joint distribution. 
As a preliminary work, we consider binary classifications
for a case study. 
Through the approach, a new upper bound is derived and shows
tighter in general than Kovalevskij's upper bound.  
The closed-form relations 
between bounds and error components are presented.
The analytical results lead to a better understanding about the
sharp conditions of bounds in terms error components.
Because classifications involve either Bayesian errors or 
non-Bayesian ones, we demonstrate the bounds comparatively
for both types of errors. 

We recognize that analytical tractability is an issue for 
the proposed approach. Fortunately, a symbolic software
tool is helpful for solving complex problems successfully 
with different semi-analytical means (such as in 
\cite{Subramanian2000}\cite{Temimi2011}). 
The semi-analytical solution used in this work 
refers to the analytical derivation of possible 
solutions but the numerical 
verification of the final solution. 

To emphasize the importance of the study, we present discussions
below from the perspective of machine learning in big-data 
classifications. We consider that binary classifications will be 
one of key techniques to implement a {\it divide-and-conquer} strategy 
for efficiently processing large quantities 
of data. 
Class-imbalance problems with extremely-skewed ratios
are mostly formed from a {\it one-against-other} division scheme
for binary classes. Researchers, of course,  
concern error types in classification performance.
The knowledge of bounds in relation to error components 
is desirable for theoretical and application purposes.

From a viewpoint of machine learning,
the bounds derived in this work 
provide a basic solution to link
learning targets between error and entropy in the related studies.
{\it Error-based learning} is more conventional
because of its compatibility with our intuitions
in daily life, such as ``{\it trial and error}''. 
Significant studies have been reported
under this category. In comparison,
{\it information-based learning} \cite{principe2010}
is relatively new and uncommon in some applications, such as classifications.
Entropy is not a well-accepted concept related to our
intuition in decision making. This is one of the reasons
why the learning target is chosen mainly based on error,
rather than on entropy. However, we consider that error is
an empirical concept, whereas entropy is theoretical and general.
In \cite{hu2012a}, we demonstrated that 
entropy can deal with both notions of {\it error} and {\it reject}
in abstaining classifications. 
Information-based learning \cite{principe2010} presents 
a promising and wider perspective for exploring and interpreting
learning mechanisms. 

When considering all sides of the issues stemming 
from machine learning studies, we believe that 
``{\it what to learn}'' is a primary problem. However, 
it seems that more investigation is focused on the issue
of ``{\it how to learn}'', which should be put as the
second-level problem.
Moreover, in comparison with the long-standing yet 
hot theme of {\it feature selection}, 
little study has been done from the perspective 
of {\it learning target selection}.
We propose that this theme should be emphasized in the study of
machine learning. Hence, the relations studied in this work 
are fundamental and crucial to the extent that researchers, 
using either error-based or entropy-based approaches, are able 
to reach a better understanding about its counterpart.  

\appendices
\begin{figure*}[!t]
\section{Maple code for deriving the lower bound}
\begin{verbatim}
> restart;                          # Clean the memory
> p2:=1-p1;e2:=Pe-e1;               # Describe the bound with respect to p1 and e1
> HT:=-p1*log[2](p1)-p2*log[2](p2); # Shannon entropy                               
> p11:=(p1-e1);p12:=e1;p22:=p2-e2;p21:=e2;            # Terms of joint probability             
> q1:=p11+p21;q2:=p12+p22;                            # Intermediate variables             
> MI:=p11*log[2](p11/q1/p1)+p12*log[2](p12/q2/p1);
> MI=MI+p22*log[2](p22/q2/(1-p1))+p21*log[2](p21/q1/(1-p1));# Mutual information  
> HTY:=(HT-MI);                                             # Conditional entropy
> HTY_dif_p1:=simplify(combine(diff(HTY,p1),ln, symbolic)); # Differential w.r.t. p1
                     /(p1 - 2 e1 + Pe) (-1 + p1 + Pe - e1)\
                   ln|------------------------------------|
                     \  (p1 - e1) (-2 e1 - 1 + p1 + Pe)   /
     HTY_dif_p1 := ----------------------------------------
                                    ln(2)                  
> HTY_dif_e1:=simplify(combine(diff(HTY,e1),ln, symbolic)); # Differential w.r.t. e1
                 /                                 2          \
                 |  (p1 - e1) (-2 e1 - 1 + p1 + Pe)  (Pe - e1)|
               ln|- ------------------------------------------|
                 |                   2                        |
                 \   (p1 - 2 e1 + Pe)  e1 (-1 + p1 + Pe - e1) /
 HTY_dif_e1 := ------------------------------------------------
                                    ln(2)                      
> solve({HTY_dif_p1=0,HTY_dif_e1=0}, {e1, p1});             # not a complete set of
                                                            # possible solutions
                 /                  2                    \ 
                 |                Pe  + e1 - Pe - 2 e1 Pe| 
                < e1 = e1, p1 = - ----------------------- >
                 |                          Pe           | 
                 \                                       / 
> E1:=solve(HTY_dif_e1, e1);     # a complete set of possible solutions when p1 is known
                       Pe (-1 + p1 + Pe)  p1 (-1 + p1 + Pe)
                 E1 := -----------------, -----------------
                           2 Pe - 1           2 p1 - 1     
> P1_a:=solve(E1[1]=e1, {p1});P1_bc:=solve(E1[2]=e1, {p1}); # a complete set of possible 
                                                            # solutions when e1 is known
                          /         2                    \ 
                          |       Pe  + e1 - Pe - 2 e1 Pe| 
                 P1_a := < p1 = - ----------------------- >
                          |                 Pe           | 
                          \                              / 
                           /                                                         (1/2)\   
                           |          1   1      1 /    2                          2\     |   
                 P1_bc := < p1 = e1 + - - - Pe + - \4 e1  - 4 e1 Pe + 1 - 2 Pe + Pe /      >, 
                           |          2   2      2                                        |   
                           \                                                              /   

                           /                                                         (1/2)\ 
                           |          1   1      1 /    2                          2\     | 
                          < p1 = e1 + - - - Pe - - \4 e1  - 4 e1 Pe + 1 - 2 Pe + Pe /      >
                           |          2   2      2                                        | 
                           \                                                              / 
> simplify(combine(simplify(eval(HTY, e1=E1[1])),ln,symbolic)); # failed to show it explicitly           
> simplify(eval(HTY, e1=E1[2]));    # Display of the lower bound function in terms of p1
                   p1 ln(p1) + ln(1 - p1) - ln(1 - p1) p1
                 - --------------------------------------
                                   ln(2)                 
> # verification of concavity of HTY by a numerical way (changing Pe and p1 arbitrarily
> Pe:=0.5;p1:=0.6;plot(HTY_graph,e1=0..Pe);                     # with the constraints)                            
\end{verbatim}

\end{figure*} 

\begin{figure*}[!t]
\section{Maple code for deriving the upper bound}
\begin{verbatim}
> restart;                                  # Clean the memory 
> HT:=-p1*log[2](p1)-p2*log[2](p2);         # Shannon entropy   
> p11:=(p1-e1);p12:=e1;p22:=p2-e2;p21:=e2;  # Terms of joint distribution  
> # To examine the HTY on two ending points for e2, i.e., e2 = 0 and e2=e
> # For derivation of the upper bound function when e2=0 
> e1:=e;e2:=0;p1:=1-p2;       
> q1:=p11+p21;q2:=p12+p22;                  # Intermediate variables 
> MI:=p11*log[2](p11/q1/p1)+p12*log[2](p12/q2/p1); # Mutual information
> MI:=MI+p22*log[2](p22/q2/(1-p1));         # Neglect one term when 0*log(0)=0  
> HTY_1:=combine(simplify(combine(simplify(HT-MI),ln,symbolic)));  
> # Display of the upper bound function when e2=e
                           /e + p2\       /e + p2\
                      p2 ln|------| + e ln|------|
                           \  p2  /       \  e   /
             HTY_1 := ----------------------------
                                 ln(2)            
> # For derivation of the upper bound function when e2=e 
> e1:=0;e2:=e;      
> q1:=p11+p21;q2:=p12+p22;                  # Intermediate variables   
> MI:=p11*log[2](p11/q1/p1);                # Neglect one term when 0*log(0)=0
> MI:=MI+p22*log[2](p22/q2/(1-p1))+p21*log[2](p21/q1/(1-p1));   
> HTY:=eval(HT-MI,p2=1-P1);     # Using P1 for p1
> HTY_2:=combine(simplify(combine(simplify(HTY),ln,symbolic)));  
> # Display of the upper bound function in terms of e and p2
                            /  P1  \       /  e   \
                      -P1 ln|------| - e ln|------|
                            \P1 + e/       \P1 + e/
             HTY_2 := -----------------------------
                                  ln(2)            
> # To calculate the difference between HTY_1 and HTY_2
> delta_HTY:=combine(simplify(HTY_1-HTY_2),ln,symbolic);                  
                     /e + p2\       /e + p2\        /  P1  \
                p2 ln|------| + e ln|------| + P1 ln|------|
                     \  p2  /       \P1 + e/        \P1 + e/
   delta_HTY := --------------------------------------------
                                   ln(2)                    
> # numerical verification of the solution to HTY below: 
> # changing p2 arbitrarily with the constraint
> # when p2<0.5, delta_HTY<0, HTY_1 is the final solution, 
> # when p2>0.5, delta_HTY>0, HTY_2 is the final solution,  
> # when p2=0.5, delta_HTY=0, both are the solutions. 
> p2:=0.4;P1:=1-p2;plot(delta_HTY,e=0..p2);
\end{verbatim}
\end{figure*}

\section*{Acknowledgments}
This work is supported in part by NSFC No. 61075051,
SPRP-CAS No. XDA06030300
for BG, and NSFC No. 60903089 for HJ. 
The previous version of this work, 
entitled ``Analytical bounds between entropy 
and error probability in binary classifications'', 
was appeared as arXiv:1205.6602v1[cs.IT] in 
May 30, 2012. Thanks to the anonymous reviewers for the 
comments and suggestions during the peer reviewing processing,
particularly for our attention to the reference [21].

%

\bibliographystyle{IEEETrans}


%





\end{document}